\documentclass[conference,letter,10pt]{IEEEtran}
\IEEEoverridecommandlockouts
\usepackage[compact]{titlesec}
\usepackage{cite}
\usepackage{flushend}
\usepackage{empheq} 
\usepackage{amsmath,amssymb,amsfonts,cases}
\usepackage{dsfont}
\usepackage{hyperref}
\usepackage{amsthm}
\usepackage{algorithmic}
\usepackage{algorithm}
\usepackage{graphicx}
\usepackage{textcomp}
\usepackage[caption=false]{subfig}
\usepackage{setspace}
\usepackage{float}
\usepackage{xcolor}
\usepackage{mathtools}
\usepackage{cite}
\usepackage{times}
\usepackage{color}
\usepackage{stackrel, url}
\usepackage{lipsum, enumitem}
\usepackage[letterpaper, left=0.625in, right=0.625in, bottom=1in, top=0.7in]{geometry}
\titlespacing*{\section}{0pt}{0.05cm}{0cm}
\titlespacing*{\subsection}{0pt}{0.05cm}{0cm}
\titlespacing*{\subsubsection}{0pt}{0.05cm}{0cm}
\theoremstyle{remark}
\newcommand\scalemath[2]{\scalebox{#1}{\mbox{\ensuremath{\displaystyle #2}}}}

\newtheorem{lemma}{\bf \emph{Lemma}}

\newtheorem{remark}{\textit{Remark}}

 \def\cF{{\mathcal{F}}}

\def\cQ{{\mathcal{Q}}}  \def\cS{{\mathcal{S}}} \def\cT{{\mathcal{T}}}
   \def\cX{{\mathcal{X}}}
\def\cY{{\mathcal{Y}}} 

\def\bff{{\mathbf{f}}}    
    
   \def\bs{{\mathbf{s}}} \def\bt{{\mathbf{t}}}
\def\bu{{\mathbf{u}}} \def\bv{{\mathbf{v}}} \def\bw{{\mathbf{w}}} \def\bx{{\mathbf{x}}} \def\by{{\mathbf{y}}}
\def\bz{{\mathbf{z}}} 

\def\bF{{\mathbf{F}}}  \def\bH{{\mathbf{H}}} \def\bI{{\mathbf{I}}}

  \def\bW{{\mathbf{W}}} \def\bX{{\mathbf{X}}}

\def\N{{\mathbb{N}}}  \def\R{{\mathbb{R}}} \def\C{{\mathbb{C}}}       \def\E{\mathbb{E}}

\def\blockdiag{\mathop{\mathrm{blk}}}

\def\argmin{\mathop{\mathrm{argmin}}}

     \def\d4{\!\!\!\!}

  \def\-{\! - \!}  \def\+{\! + \!}  \def\={\! = \!}  \def\>{\! > \!}

\newcommand{\bef}{\begin{figure}}
\newcommand{\eef}{\end{figure}}
\newcommand{\beq}{\begin{eqnarray}}
\newcommand{\eeq}{\end{eqnarray}}

\def\sumK{\sum_{k=1}^{K}}

\def\sumN{\sum_{n=1}^{N}}
\def\sumM{\sum_{m=1}^{M}}
\def\sumP{\sum_{p=1}^{P}}

\begin{document}
\title{ 
No Analog Combiner TTD-based Hybrid Precoding for Multi-User Sub-THz Communications
}
\author{
    \IEEEauthorblockN{Dang Qua Nguyen\IEEEauthorrefmark{1}, Alexei Ashikhmin\IEEEauthorrefmark{2}, Hong Yang\IEEEauthorrefmark{2}, and Taejoon Kim\IEEEauthorrefmark{1}}
    \IEEEauthorblockA{\IEEEauthorrefmark{1}\textit{Department of Electrical Engineering and Computer Science, The University of Kansas, Lawrence, KS 66045 USA}\\ \IEEEauthorrefmark{2}\textit{Nokia Bell Labs, Murray Hill, NJ 07974 USA}}
    Email: \IEEEauthorrefmark{1}\{quand, taejoonkim\}@ku.edu \& \IEEEauthorrefmark{2}\{alexei.ashikhmin, h.yang\}@nokia-bell-labs.com
   \thanks{ The work of D. Q. Nguyen and T. Kim was supported in part by the National Science Foundation (NSF) under Grant CNS1955561, CNS2212565, and CNS2225577, and the Office of Naval Research (ONR) under Grant N000142112472.}
}

\maketitle
\begin{abstract}
We address the design and optimization of real-world-suitable hybrid precoders for multi-user wideband sub-terahertz (sub-THz) communications. 
We note that the conventional fully connected true-time delay (TTD)-based architecture is impractical because there is no room for the required large number of analog signal combiners in the circuit board. 
Additionally, analog signal combiners incur significant signal power loss. 
These limitations are often overlooked in sub-THz research.  
To overcome these issues, we study a non-overlapping subarray architecture that eliminates the need for analog combiners. 
We extend the conventional single-user assumption by formulating an optimization problem to maximize the minimum data rate for simultaneously served users. 
This complex optimization problem is divided into two sub-problems. 
The first sub-problem aims to ensure a fair subarray allocation for all users and is solved via a continuous domain relaxation technique. 
The second sub-problem deals with practical TTD device constraints on range and resolution to maximize the subarray gain and is resolved by shifting to the phase domain. 
Our simulation results highlight significant performance gain for our real-world-ready TTD-based hybrid precoders.\looseness=-1
\end{abstract}
\begin{IEEEkeywords}\noindent 
Sub-THz communication, multi-user, true-time delay, joint delay and phase precoding, non-overlapping subarrays, and hybrid massive MIMO. \looseness=-1
\end{IEEEkeywords}
\section{Introduction}
Sub-terahertz (sub-THz) band (90-300 GHz) communication is a promising technology for achieving hundreds of Gbps data rates in the sixth-generation (6G) wireless systems \cite{Giordani2020}. 
For sub-THz communications, a hybrid massive multiple-input multiple-output (MIMO) system combined with orthogonal frequency division multiplexing (OFDM) has been popularly studied \cite{Yuan2022}.
The considerably wide bandwidth and large array sizes of the hybrid massive MIMO OFDM system, however, cause beam squint, which deviates the beam direction across OFDM subcarriers, severely deteriorating array gain \cite{Wang2019, Qua2022}. 
Recently, the true-time delay (TTD)-based hybrid precoding has emerged as a method to address the beam squint \cite{dai2021, Gao2021, Matthaiou2021, Qua2023, Najjar2023}.
Unlike phase shifter (PS)-based hybrid precoding, which applies frequency-independent phase adjustments, TTD-based hybrid precoding introduces frequency-dependent phase modifications to cope with the beam squint.
 Previous works \cite{dai2021, Gao2021, Matthaiou2021, Qua2023, Najjar2023} mainly studied single-user TTD-PS precoders for the fully-connected architecture, where each data stream uses the whole antenna array. 
We note that this fully connected TTD-PS architecture is impractical due to the significantly increased hardware complexity and power loss \cite{Cabric2018CombinerLoss}. 
These practical limitations are often neglected in sub-THz hybrid precoders design research. 
To reduce the hardware complexity of the fully-connected architecture, a non-overlapping subarrays architecture is introduced for the downlink multi-user TTD-based hybrid precoding systems. 
One data stream is conventionally mapped to a subarray dedicated to a user \cite{Najjar2022}, an issue remaining is then when the number of subarrays is larger than the number of users. 
Moreover, most of the prior works assumed either unbounded \cite{dai2021, Gao2021, Matthaiou2021, Najjar2022} or infinite resolution \cite{dai2021, Gao2021, Matthaiou2021, Qua2023, Najjar2022}  TTD values, which are impractical \cite{Lin2022}. 
Rate fairness among users in downlink multi-user networks is a pragmatic objective that has been widely studied in massive MIMO literature \cite{ZQLou2016, A.Hien2017, Hong2017, TJ2021}.
These works use power control and/or precoder optimization to maximize the minimum rate or minimum signal-to-interference-plus-noise-ratio (SINR) performance.
The underlying assumption is a fully-digital array architecture, which is unsuitable for wideband sub-THz massive MIMO systems \cite{dai2021}. 
Addressing the challenges in designing and optimizing real-world-suitable TTD-based hybrid precoders for fair multi-user wideband sub-THz communications is still nascent. 
This paper investigates the non-overlapping TTD-based subarray architecture and aims to achieve rate fairness in downlink multi-user wideband sub-THz massive MIMO OFDM systems.
The rate fairness is achieved through the design and optimization of subarray allocation and hybrid precoding methods.  
This problem is divided into two sub-problems.
The first sub-problem deals with a fair subarray allocation to users, which addresses the case when there is more number of subarrays than the number of users and is solved by continuous domain relaxation techniques. 
The second sub-problem addresses the practical TTD device constraints on range and resolution to maximize the subarray gain at each subarray which is resolved by a phase domain transformation approach.
Our simulations verify the advantages of the proposed framework compared to the existing approach \cite{Najjar2022}. 
    \noindent \textit{Notation:} A bold lower case letter $\bx$ is a column vector and a bold upper case letter $\bX$ is a matrix. 
    $\bX^T$, $\bX^{\dag}$, $\bX(i,j)$, $\|\bx\|$, and $|x|$ are, respectively, the transpose, conjugate transpose, $i$th row and $j$th column entry of $\bX$, $2$-norm of $\bx$, and modulus of $x\in \C$. 
    $\blockdiag(\bX_1, \bX_2, \dots, \bX_N)$ is an $ NN_1 \times NN_2$ block diagonal matrix such that its main-diagonal blocks contain $\bX_n \in \C^{N_1\times N_2}$, for $n=1,\dots,N$, and all off-diagonal blocks are zero.
    $\bI_{n}$ denotes the $n\times n$ identity matrix.
    $e^{j\bx}$ denotes the vector $[e^{jx_1}~e^{jx_2} ~\dots ~e^{jx_n}]^T \in \C^{n \times 1}$ obtained by applying $e^j$ element-wise to $\bx = [x_1~x_2\dots~x_n]^T\in \R^{n \times 1}$. 
    \looseness=-1
\section{Channel Model}
We consider a downlink multi-user wideband sub-THz massive MIMO OFDM channel. 
The transmitter is equipped with an $N_t$-element uniform linear array (ULA) and communicates with $N$ distributed users. 
Each user is equipped with an $N_r$-element ULA. 
We let $f_c$, $B$, and $K$ be, respectively, the central (carrier) frequency, bandwidth of the OFDM system, and the number of OFDM subcarriers (an odd integer so that there is a central subcarrier).
The $k$th subcarrier frequency is given by \looseness=-1
\vspace{-0.1cm}
\begin{equation}
     \label{eqfk}
     \scalemath{0.9}{f_k = f_c + \frac{B}{K}\Big(k-1-\frac{K-1}{2}\Big),\text{ for } k = 1,\dots,K.}
\vspace{-0.1cm}
\end{equation}

Due to the sparsity of the sub-THz channel \cite{Han2023}, we assume that there is one path from the transmitter to each user. 
Assuming the far-field communication, the channel between the transmitter and the $n$th user at the $k$th subcarrier is modeled by \looseness=-1 
\vspace{-0.1cm}
\begin{equation}
    \label{eq:channel}
    \scalemath{1}{\bH_{k,n} = \sqrt{N_tN_r} \alpha_{k,n} e^{-j 2 \pi f_k \tau_{n}} \bv_{k,n}\bu^{\dag}_{k,n} \in \C^{N_r \times N_t},}
    \vspace{-0.1cm}
\end{equation}\normalsize 
where $\alpha_{k,n} \in \C$ and  $\tau_{n} \in \R$ are the channel path gain and propagation delay, respectively.
The molecular absorption property of sub-THz propagation waves is characterized by the frequency-dependent gain $\alpha_{k,n}$ \cite{Jornet2011, Han2015}.  
The $\bv_{k,n} \in \C^{N_r \times 1}$ and $\bu_{k,n} \in \C^{N_t \times 1}$ are the receive and transmit array response vectors, respectively.
The $i$th entry of $\bu_{k,n}$ is given by $\bu_{k,n}(i) = \frac{1}{\sqrt{N_t}}e^{-j \pi \frac{2d f_k}{c} (i-1) \sin(\Psi_{n})}$, where the $\Psi_{n}$ is the angle of departure (AoD), and c is the speed of light.
We define $\psi_{k,n} = \frac{2d f_k}{c}\sin(\Psi_{n})$ as the spatial direction from the transmitter to the $n$th user at the $k$th subcarrier.
Assuming the half-wavelength antenna spacing, i.e., $d = \frac{c}{2f_c}$, the spatial direction at the central frequency is reduced to $\psi_{n} = \sin(\Psi_{n})$.
Defining $\xi_k = \frac{f_k}{f_c}$ yields 
$\psi_{k,n} = \xi_k \psi_{n}$ and \looseness=-1
\vspace{-0.1cm}
\begin{equation}
\label{eq:xik}
\scalemath{0.9}{\xi_k = 1 + \frac{B}{f_c}\Big(\frac{k-1-\frac{K-1}{2}}{K}\Big),}
 \vspace{-0.1cm}
\end{equation}\normalsize 
where \eqref{eq:xik} follows from the definition of $f_k$ in \eqref{eqfk}. 
As a result, the transmit array response vector $\bu_{k,n}$ in \eqref{eq:channel} is 
\vspace{-0.1cm}
\begin{equation}
    \label{eq:arrayRV}
    \scalemath{0.95}{\bu_{k,n} = \frac{1}{\sqrt{N_t}}\Big[1 ~e^{j\pi\psi_{k,n}}~\dots ~e^{j\pi (N_t-1)\psi_{k,n}} \Big]^{\dagger}  \in  \C^{N_t \times 1}.}
\vspace{-0.1cm}
\end{equation}
The receive array response vector $\bv_{k,n}$ is defined similarly,
\vspace{-0.1cm}
\begin{equation}
    \label{eq:arrayRV2}
    \scalemath{0.95}{\bv_{k,n} = \frac{1}{\sqrt{N_r}}\Big[1 ~e^{j\pi\phi_{k,n}}~\dots ~e^{j\pi (N_r-1)\phi_{k,n}} \Big]^{\dagger}  \in  \C^{N_r \times 1},}
\vspace{-0.1cm}
\end{equation}
where $\phi_{k,n} = \xi_k\phi_n$, $\phi_n = \sin(\Phi_n)$, and $\Phi_n$ is the angle of arrival (AoA) at the $n$th user.
\section{System Model and Problem Formulation}
   \begin{figure}[htp]
    \centering     \includegraphics[width=0.48\textwidth]{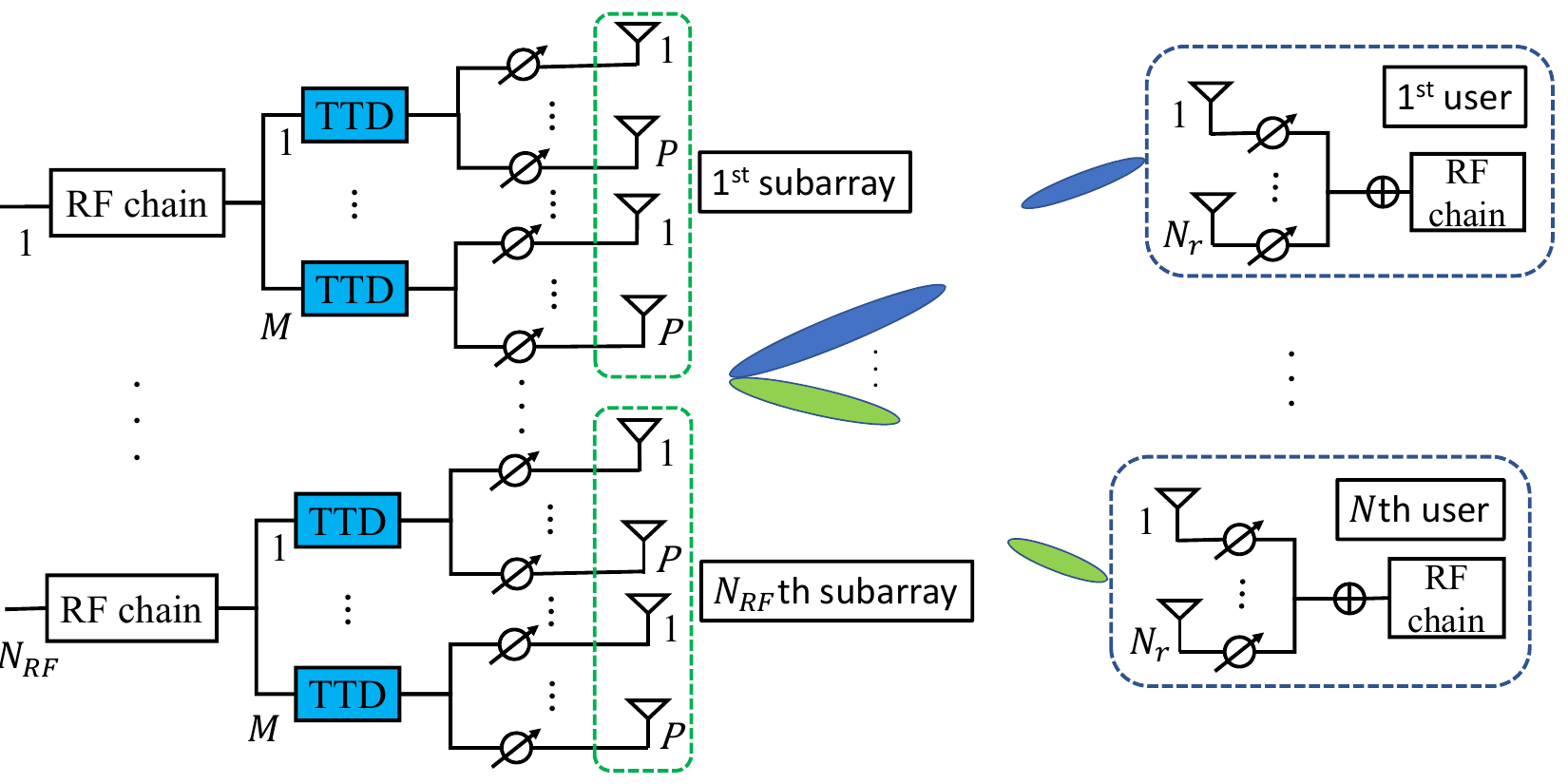}
    \caption{Non-overlapping subarray TTD-based hybrid precoding for multi-user massive MIMO}
    \label{Fig1}
    \vspace{-0.5cm}
    \end{figure} 
We consider a non-overlapping subarray TTD-based hybrid precoding architecture, as shown in Fig.~\ref{Fig1}, where the transmit ULA is fed by $N_{RF}$ radio frequency (RF) chains to transmit $N_s$ data streams ($N_{RF} \geq N_s$) simultaneously to  $N$ single-RF chain users.
In what follows, we assume $N_s = N$ and use $N$ to denote both the number of users and the number of data streams.  
Each RF chain corresponds to a subarray, as illustrated in Fig.~\ref{Fig1}.
Particularly, each RF chain drives $M$ TTDs and each TTD is connected to $P$ PSs, i.e., $N_t = N_{RF}MP$.
We note that analog combiners are no longer required with this subarray architecture. This is in contrast with the prior works like \cite{dai2021, Gao2021, Matthaiou2021, Qua2023, Najjar2023} that need the use of analog combiners for precoding. 
A TTD delays the transmit signal at the $k$th subcarrier by $t$ in the time domain, equivalent to $-2\pi f_kt$ phase rotation in the frequency domain.
Unlike prior works \cite{Najjar2022,dai2021, Gao2021, Matthaiou2021, Qua2023}, we assume that the time delay $t$ is chosen in a finite set $\cT$ such that $t \in \cT = \{0,\tau, 2\tau, \dots, (Q-1)\tau\}$, where $\tau$ is the step size and $Q$ is the quantization level. 
Due to the non-overlapping subarray in Fig.~\ref{Fig1}, the transmit array response vector $\bu_{k,n}$ in \eqref{eq:arrayRV} can be rewritten as $\scalemath{0.9}{\bu_{k,n} = [\bu_{k,n,1}^T~\bu_{k,n,2}^T~ \dots~ \bu^T_{k,n,N_{RF}}]^T}$, where $\scalemath{0.9}{\bu_{k,n,l} \in \C^{MP \times 1}}$ is the $l$th subarray response vector, for $l=1,\dots,N_{RF}$, and is given by
\vspace{-0.1cm}
\begin{equation}
    \label{eq:subarrayRV}
    \scalemath{0.95}{\d4\bu_{k,n,l} \!=\! \frac{1}{\sqrt{N_t}}e^{-j \pi (l-1)MP \psi_{k,n}}\!\Big[ 1~e^{j\pi\psi_{k,n}} \dots e^{j\pi (MP-1)\psi_{k,n}} \Big]^{\dagger}\!\!.}
\vspace{-0.1cm}
\end{equation}\normalsize
To facilitate the subarray allocation to multi-users, we define a subset of subarray indices that are used to serve the $n$th user as $\cS_n$, where $\bigcup_{n=1}^{N}\cS_n = \{1,2,\dots,N_{RF}\}$, $\cS_n \neq \emptyset$, and $\cS_{n} \cap \cS_{n'} = \emptyset$, $\forall n \neq n'$.
To reduce design complexity, we assume a consecutive subarray allocation to each user such that $\cS_1 = \{1,2,\dots,|\cS_1|\}$, $\cS_2 = \{|\cS_1|+1,\dots, |\cS_1|+|\cS_2|\}, \dots, \cS_N = \{1+\sum_{n=1}^{N-1}|\cS_{n}|,\dots,N_{RF}\}$, where $|\cS_n|$ denotes the cardinality of $\cS_n$. 
This simplification turns the problem of subarray allocation into the problem of determining $\{  |\cS_n|\}_{n=1}^N$. \looseness=-1
The received signal of the $n$th user at the $k$th subcarrier is then given by \looseness = -1
\vspace{-0.1cm}
\begin{equation}
\label{eq:signal}
    \scalemath{0.95}{\by_{k,n} = \sqrt{\rho} \bH_{k,n}\bF_1\bF_{2,k}\bW_{k}\bs_k + \bz_{k,n},} 
\vspace{-0.1cm}
\end{equation}  
where $\rho$ is the transmit power, $\bH_{k,n} \in \C^{N_r \times N_{t}}$ is defined in \eqref{eq:channel}, $\bF_1 \in \C^{N_{t} \times N_{RF}M} $ is the PS precoder, $\bF_{2,k} \in \C^{N_{RF}M \times N_{RF}}$ is the TTD precoder, $ \bW_{k} \in \C^{N_{RF} \times N}$ is the digital precoder, $\bz_{k,n} \in \C^{N_r \times 1}$ is the complex Gaussian random noise vector with zero mean and covariance $\bI_{N_r}$, and $\bs_{k} = [s_{k,1} \dots s_{k,N}]^T \in \C^{N \times 1}$ is the vector of transmit symbols such that $\E[\bs_k\bs^{\dag}_k] = \frac{1}{N}\bI_{N}$. 
Specifically, the $n$th symbol $s_{k,n}$ is mapped to the $n$th user.
The PS precoder $\bF_1$ is defined as 
\vspace{-0.1cm}
\begin{equation}
        \label{eq:F1}
        \scalemath{0.95}{\bF_1 = \blockdiag(\bX_1, \dots, \bX_{N_{RF}}),} 
\vspace{-0.1cm}
\end{equation}
where $\bX_l =  \blockdiag (e^{j \pi \bx_{l,1}}, \dots, e^{j \pi \bx_{l,M}}) \in \C^{MP \times M}$ is the $l$th PS sub-precoder and $\bx_{l,m} = [x_{l,m,1} \dots x_{l,m,P}]^T \in \R^{P \times 1}$ is the phase shifter sub-vector connected to the $m$th TTD of the $l$th RF chain.
We denote $\cX_{N_t,N_{RF},M}$ as the set of all $\bF_1$ defined in \eqref{eq:F1}. 
The TTD precoder $\bF_{2,k}$ is given by \sloppy
\vspace{-0.1cm}
\begin{equation}
    \label{eq:F2k}
    \scalemath{0.95}{\bF_{2,k} = \blockdiag (e^{-j2\pi f_k \bt_{1}}, \dots , e^{-j2\pi f_k \bt_{N_{RF}}}),}
\vspace{-0.1cm}
\end{equation} 
where $e^{-j2\pi f_k \bt_{l}} \in \C^{M \times 1}$ is the $l$th TTD sub-precoder and $\bt_{l} = [t_{l,1} \dots t_{l,M}]^T \in \cT^{M \times 1}$ is the delay values of the $M$ TTDs connected to the $l$th RF chain.
We denote $\cF_{k,N_{RF},M}$ as the set of all $\bF_{2,k}$ defined in \eqref{eq:F2k}. 
Combining \eqref{eq:F1} and \eqref{eq:F2k}, we have the analog precoder at the $k$th subcarrier 
\vspace{-0.1cm}
\begin{equation}
    \label{eq:Fk}
    \bF_1\bF_{2,k} = \blockdiag(\bff_{k,1}, \dots, \bff_{k,N_{RF}}),
    \vspace{-0.1cm}
\end{equation}
where $\bff_{k,l} \in \C^{MP \times 1}$ is the $l$th analog sub-precoder and is given by 
\vspace{-0.1cm}
 \begin{equation}
 \label{eq:subprecoder}
 \scalemath{0.95}{\bff_{k,l} = \bX_{l}e^{-j2\pi f_k \bt_{l}}, \forall k, l.}
 \vspace{-0.1cm}
 \end{equation}
The digital precoders $\{\bW_k\}_{k=1}^{K}$ in \eqref{eq:signal} are normalized to satisfy the transmit power constraint $\scalemath{0.9}{\|\bF_1\bF_{2,k}\bW_k\|^2_F = N}$, $\forall k$.\sloppy 
Taking into account the inter-user interference, the received signal at the $n$th user can be rewritten as
\vspace{-0.2cm}
\begin{multline}
    \label{eq:signal2}
    \d4\scalemath{0.9}{\by_{k,n} = \sqrt{\rho}\bH_{k,n} \bF_1\bF_{2,k}\bw_{k,n}s_{k,n}} \\ + \scalemath{0.9}{\!\sum_{n' \neq n}\!\! \sqrt{\rho}\bH_{k,n} \bF_1\bF_{2,k} \bw_{k,n'} s_{k,n'} 
    + \bz_{k,n}} \in \C^{N_r \times 1}, 
\end{multline}
where the second term on the right-hand side (r.h.s) of \eqref{eq:signal2} represents the inter-user interference and $\bw_{k,n} \in \C^{N_{RF} \times 1}$ denotes the $n$th column of $\bW_k$. 
Based on \eqref{eq:signal2}, the achievable rate of the $n$th user is given by \looseness=-1
 \vspace{-0.1cm}
\begin{equation}
    \label{eq:Rate}
    \scalemath{0.9}{R_n = \sumK \log_2 \Big( 1 + \frac{\frac{\rho}{N}\|\bH_{k,n} \bF_1\bF_{2,k}\bw_{k,n}\|^2}{ \sum_{n' \neq n} \frac{\rho}{N}\|\bH_{k,n} \bF_1\bF_{2,k} \bw_{k,n'}\|^2+1} \Big).}
 \vspace{-0.1cm}
\end{equation}
The objective is to design the subarray allocation $\{|\cS_n|\}$, and precoders $\bF_1$, $\{\bF_{2,k}\}$, and $\{\bW_k\}$ that maximize the minimum $R_n$ in \eqref{eq:Rate}, ensuring a fair rate performance for all users: \looseness=-1  
\vspace{-0.2cm}
\begin{subequations}
\label{eq:globalProblem}
 \beq
    \d4\d4\!\d4\d4\d4\max_{\{|\cS_n|\},\bF_1, \{\bF_{2,k}\}, \{\bW_k\}} \! \min_{n}\!\! && \d4\d4 R_n, \label{eq:GPa}\\[0.5pt]
    \text{subject to}\!\! && \d4\d4 \scalemath{0.9}{\sumN |\cS_n| = N_{RF}, |\cS_n| \in \N, \forall n,} \label{eq:GPb} \\
    && \d4\d4 \bF_1 \! \in \! \cX_{N_t,N_{RF},M},   \label{eq:GPc}\\ 
    && \d4\d4 \|\bF_1\bF_{2,k}\bW_k\|_F^2 = N, \label{eq:GPd}\\
    && \d4\d4 \bF_{2,k} \! \in \! \cF_{k,N_{RF},M},\label{eq:GPe} 
\eeq 
\end{subequations}where the constraint in \eqref{eq:GPb} corresponds to the total number of subarrays and $\N$ denotes the set of natural numbers. 
The constraints in \eqref{eq:GPc} is due to the structure of the PS precoder $\bF_1$ in \eqref{eq:F1}, \eqref{eq:GPd} indicates the normalized power of the precoders, and \eqref{eq:GPe} is due to the structure of $
\bF_{2,k}$ in \eqref{eq:F2k}.\looseness=-1
The problem \eqref{eq:globalProblem} is not directly solvable due to the coupling between the precoders and the non-convex constraints \eqref{eq:GPb}--\eqref{eq:GPe}.
In this paper, this difficulty is approached by devising efficient algorithms through decomposition of the original max-min problem in \eqref{eq:globalProblem} into two sub-problems: (i) user subarrays allocation problem and (ii) joint PS and TTD sub-precoders design problem. 
Subsequently, the digital precorders $\{\bW_k\}$ are designed by standard approaches \cite{Sayed2007, Alkhateeb20152}.
In the following sections, we focus on the design and optimization of the analog precoders and fair subarray allocation. 
\looseness=-1 

\section{Fair Subarray Allocation}\label{SecIV}
In this section, we propose a max-min optimization approach for allocating the number of subarrays $\{|\cS_n|\}$ to each user while fixing the precoders $\bF_1$, $\{\bF_{2,k}\}$, and $\{\bW_k\}$.
To incorporate a tractable objective function, we consider an asymptotic upper bound of $R_n$ in \eqref{eq:Rate} as the number of antennas at each subarray tends to infinity ($MP \to \infty$).
To this end, the analog sub-precoders $\{ \bff_{k,l} \}$ in \eqref{eq:subprecoder} are assumed to be ideal analog sub-precoders (IASP) to maximize the respective subarray gains \cite{dai2021,Qua2022}, which are designed to perfectly match the subarray response vectors for $l \in \cS_n$ such that 
\vspace{-0.1cm}
\begin{equation}
\label{eq:assumption1}
    \bff_{k,l} = \bu_{k,n,l}, \forall k, 
\vspace{-0.1cm}
\end{equation}
where $ \bu_{k,n,l}$ is defined in \eqref{eq:subarrayRV}.
Furthermore, \eqref{eq:GPd} implies that the power of each column of $\bW_k$ must be bounded:
\vspace{-0.1cm}
\begin{equation}
\label{eq:assumption2}
    \max_{k,n}\|\bw_{k,n}\|^2 = \omega < \infty.
\vspace{-0.1cm}
\end{equation} 
Given the latter assumptions, an assymptotic upper bound of $R_n$ in \eqref{eq:Rate} is presented in Lemma~\ref{lm:RateUB}.   
\begin{lemma}
    \label{lm:RateUB}
    The following bound holds, $\lim_{MP \rightarrow \infty}R_n \leq \frac{\rho\omega N_r N_t}{\ln(2)N}\frac{|\cS_n|}{N_{RF}}\sumK|\alpha_{k,n}|^2$ with probability 1, for $n=1,\dots,N$.
\end{lemma}
\begin{proof}
    See Appendix~\ref{appendixC}.
\end{proof}
Incorporating the asymptotic upper bound in Lemma~\ref{lm:RateUB}, the max-min fairness subarray allocation problem is given by
\vspace{-0.2cm}
\begin{subequations}
    \label{subarray_allocation}
    \beq   
    \d4\d4\max_{\{|\cS_n|\}}\min_{n} && \d4 \scalemath{0.9}{\widetilde{\alpha}_{n}|\cS_n|},\label{eq:SAa}\\
    \text{subject to} &&\d4  \scalemath{0.9}{\sumN |\cS_n| = N_{RF}, |\cS_n| \in \N, \forall n}, \label{eq:SAc}
    \eeq 
\end{subequations}
where $\widetilde{\alpha}_n = \sumK |\alpha_{k,n}|^2$ is the sum of the channel gains at the $n$th user and the constraint in \eqref{eq:SAc} corresponds to \eqref{eq:GPb}.
The integer programing in \eqref{subarray_allocation} can be conventionally addressed by a branch, bound, and reduced \cite{Lawler1966, TJ2015} and cutting plane algorithms \cite{Gomory1961cutting}. 
Albeit straightforward, the worst case complexity of these algorithms is combinatorial with $N$.
Instead, we address the discrete constraint in \eqref{eq:SAc} via a continuous domain relaxation technique:   \looseness=-1
\vspace{-0.2cm}
\begin{subequations}
    \label{subarray_allocation2}
    \beq 
    \d4\d4\max_{\{|\cS_n|\}} ~~ \min_{n} &&\d4 \scalemath{0.9}{\widetilde{\alpha}_{n}|\cS_n|},\label{eq:SAa2}\\
    \text{subject to} && \d4  \scalemath{0.9}{\sumN |\cS_n| = N_{RF}, |\cS_n| > 0, \forall n}. \label{eq:SAb2}
    \eeq 
\end{subequations}
The following lemma characterizes the closed-form solution to \eqref{subarray_allocation2}. \looseness=-1
\begin{lemma}
\label{lemma:optimal_SA2}
The closed-form solution to  \eqref{subarray_allocation2} is given by 
$|\cS_n| = \frac{N_{RF}/\widetilde{\alpha}_n}{\sumN 1/\widetilde{\alpha}_n}, \forall n.$
\end{lemma}
\begin{proof}
    See Appendix~\ref{appendixA}. 
\end{proof}
\begin{remark}\label{rmk1}
    Lemma~\ref{lemma:optimal_SA2} implies that we should allocate more subarrays to the users that experience more severe path-loss to provide a fair subarray allocation among users. 
\end{remark}
\noindent Discretizing the solution in Lemma~\ref{lemma:optimal_SA2} yields an approximate solution to \eqref{subarray_allocation}, which is given by 
\vspace{-0.2cm}
    \begin{subequations}
    \label{eq:closeform}
        \beq
        \scalemath{0.9}{|\cS_n|^{\star}} &=& \scalemath{0.9}{\Big\lfloor \frac{N_{RF}/\widetilde{\alpha}_n}{\sumN 1/\widetilde{\alpha}_n} \Big\rfloor, \text{for }n=1,\dots,N-1,} \label{eq:closeforma}\\
        \scalemath{0.9}{|\cS_N|^{\star}} &=& \scalemath{0.9}{N_{RF} - \sum_{n=1}^{N-1} | \cS_n |^{\star},} \label{eq:closeformb}
        \eeq 
    \end{subequations}
    where \eqref{eq:closeformb} is to ensure the first constraint in \eqref{eq:SAc} and $\lfloor x \rfloor$ is the largest integer that is less than or equal to $x$.\looseness=-1      
\section{Joint PS and TTD Sub-precoders Design}\label{SecV}
Given the subarray allocation in \eqref{eq:closeform}, we now jointly design PS and TTD sub-precoders at the $l$th subarray to achieve the IASP in  \eqref{eq:assumption1} for maximizing the subarray gain.
However, solving for the $K$ coupled equations in \eqref{eq:assumption1} is ill-posed. 
Hence, we attempt to best approximate $\bu_{n,k,l}$ with the construction of $\bff_{k,l}$ in \eqref{eq:subprecoder} as presented,\looseness=-1 
\vspace{-0.1cm}
\begin{subequations}
    \label{eq:opt1}
    \beq
    \min_{\bX_l, \bt_l} && \d4\d4 \scalemath{1}{\sumK \|\bX_l e^{-j2 \pi f_k \bt_l} - \bu_{k,n,l}\|^2}, \label{eq:opt1_0}\\ 
    \text{subject to} 
    &&\d4\d4 \scalemath{1}{|\bX_{l}(i,j)| \in \Big\{0,\frac{1}{\sqrt{N_t}} \Big\}, \forall i, j,}    \label{eq:opt1a}\\
    &&\d4\d4 \scalemath{1}{t_{l,m} \in \cT, \forall m,}    \label{eq:opt1b}
    \eeq 
\end{subequations}
where the constraints in \eqref{eq:opt1a} and \eqref{eq:opt1b} are restatements of \eqref{eq:GPc} and  \eqref{eq:GPe}, respectively. 
These non-convex constraints in conjunction with the coupling between $\bX_l$ and $e^{-j2 \pi f_k \bt_l}$ make the problem in \eqref{eq:opt1} still difficult to solve. 
In conventional approaches \cite{ayach2014, Hadi2016}, the problem in \eqref{eq:opt1} is addressed via alternating optimization.
Unlike the previous work \cite{ayach2014, Hadi2016}, we transform \eqref{eq:opt1} to a problem in the phase domain to remove the non-convex constraint in \eqref{eq:opt1a}, and then the discrete constraint in \eqref{eq:opt1b} is dealt by an alternating optimization method.
From \eqref{eq:subarrayRV}, the objective function in \eqref{eq:opt1_0} can be rewritten as
\vspace{-0.2cm}
\begin{equation}
    \label{eq:obj1}
    \d4\scalemath{1}{\frac{1}{N_t}\sumK \!\sumM \! \sumP|e^{-j \pi 2f_c\xi_k t_{l,m}} e^{j \pi x_{l,m,p}}-e^{-j \pi  \xi_k\gamma_{n,l,m,p}}|^2,}
\end{equation} 
where $\scalemath{0.9}{\gamma_{n,l,m,p} = ((l-1)MP + (m-1)P+p-1)\psi_{n}}$ and $\xi_k$ is in \eqref{eq:xik}.
We recall a lemma in the previous work \cite[Lemma~1]{Qua2023} that allows us to transform the problem in \eqref{eq:opt1} into a relaxed problem in the phase domain.
     \begin{lemma}(\cite[Lemma~1]{Qua2023})
     \label{lm2}
      For a given $x\in \R$, and a closed set $\cY \subseteq  (x-\pi, x+\pi]$, $\argmin_{y\in \cY}|e^{jx}-e^{jy}|=\argmin_{y\in \cY}|x-y|$.
    \end{lemma}
\noindent Incorporating Lemma~\ref{lm2} into \eqref{eq:obj1} transforms \eqref{eq:opt1} to a relaxed problem:
\vspace{-0.1cm}
\begin{equation}
\begin{aligned}
    \!\!\!\!\scalemath{0.95}{\min_{\{x_{l,m,p}\}, \{t_{l,m}\}}}&~ \scalemath{0.95}{\sumK\!\sumM\!\sumP (\!-2f_c\xi_k t_{l,m} \!+\! x_{l,m,p} \!+\! \xi_k \gamma_{n,l,m,p})^2}\!\!,\\ 
    \text{subject to} &~\scalemath{1}{t_{l,m} \in \cT, \forall m},\nonumber 
\end{aligned}    
\end{equation}
which can be decomposed into $M$ separable problems, for $m=1,\dots,M$:
\vspace{-0.3cm}
%
\begin{equation}
    \begin{aligned}
   \label{eq:opt3}
    \d4\d4 \scalemath{1}{\min_{\{x_{l,m,p}\}, t_{l,m}}} &~~ \scalemath{1}{\sumK \! \sumP (-2f_c\xi_k t_{l,m} + x_{l,m,p} + \xi_k \gamma_{n,l,m,p})^2}\!,\\ 
    \text{subject to} 
    &~~ \scalemath{1}{t_{l,m} \in \cT}.  
    \end{aligned}
\end{equation}
To circumvent the difficulty in dealing with the discrete constraint in \eqref{eq:opt3} and the coupling between $t_{l,m}$ and $\{x_{l,m,p}\}$, we devise an alternating optimization algorithm.  
In particular, optimizing one parameter by fixing the other parameter in \eqref{eq:opt3} addresses the coupling issue. 
First, $t_{l,m}$ is optimized while fixing $\{x_{l,m,p}\}$ and temporarily neglecting the constraint in \eqref{eq:opt3}: \looseness=-1
\vspace{-0.2cm}
\begin{equation}
    \label{eq:opt5}
\min_{t_{l,m}} ~~ \sumK \sumP    \scalemath{1}{(-2f_c\xi_k t_{l,m} + x_{l,m,p} + \xi_k \gamma_{n,l,m,p})^2.}
\vspace{-0.2cm}
\end{equation} 
Since \eqref{eq:opt5} is convex, the optimal solution to \eqref{eq:opt5} is found by setting the first-order derivative of \eqref{eq:opt5} with respect to $t_{l,m}$ to zero, i.e., $\sumK \sumP -2f_c\xi_k(-2f_c\xi_k t_{l,m} $ $+ x_{l,m,p} + \xi_k \gamma_{n,l,m,p}) = 0$. 
From \eqref{eq:xik}, one can confirm that $\sumK \xi_{k}^2 = \Gamma K$, where $\Gamma = 1+\frac{B^2(K^2-1)}{12f_c^2K^2}$, leading to the solution to \eqref{eq:opt5},
\vspace{-0.1cm}
\begin{equation}\label{vartheta_lm}
    \scalemath{1}{t_{l,m} = \frac{1}{2f_cP} \Big(\frac{1}{\Gamma}\sumP x_{l,m,p} + \sumP \gamma_{n,l,m,p}\Big).} 
\vspace{-0.1cm}
\end{equation}
Next, $\{x_{l,m,p}\}$ are optimized while fixing $t_{l,m}$: 
\vspace{-0.1cm} 
\begin{equation}
    \label{eq:opt4}
\min_{\{x_{l,m,p}\}} ~~ \scalemath{1}{\sumK \sumP (-2f_c\xi_k t_{l,m} + x_{l,m,p} + \xi_k \gamma_{n,l,m,p})^2.}
\vspace{-0.1cm} 
\end{equation}
Similarly to \eqref{eq:opt5}, \eqref{eq:opt4} is also convex and its solution is attained by equating its first-order derivative to be zero, $\sumK \scalemath{0.92}{2(-2f_c\xi_k t_{l,m} + x_{l,m,p} + \xi_k \gamma_{n,l,m,p}) = 0}$, $\forall p$. 
From \eqref{eq:xik}, it is not difficult to verify that $\sumK \scalemath{0.9}{\xi_k = K}$, implying the solution to \eqref{eq:opt4},\looseness=-1
\vspace{-0.2cm}
\begin{equation}\label{eq:x_lmp}
    \scalemath{1}{x_{l,m,p} =  2f_ct_{l,m} - \gamma_{n,l,m,p}, \forall p.}
\vspace{-0.5cm}    
\end{equation}
\begin{algorithm}[htp]
   \caption{Alternating optimization of PS and TTD values}
   \label{algorithm1}
\begin{algorithmic}[1]
 \renewcommand{\algorithmicensure}{\textbf{Output:}}
 \REQUIRE $\{\xi_k\}$, $\{\gamma_{n,l,m,p}\}$, $N_{iter}$, $\cT$, and $\epsilon$  
 \ENSURE $\{t^{\star}_{l,m}\}$ and $\{x^{\star}_{l,m,p}\}$  
 \STATE{\textbf{Initialization:} $\scalemath{0.9}{t^{(1)}_{l,m}=0, \forall m}$, $\scalemath{0.9}{x^{(1)}_{l,m,p}=0}$, $\forall m, p$,  and $i = 1$ \label{step1}} 
    \FOR{$m = 1:M$}
    \WHILE{$i < N_{iter}$}   
    \STATE{$t^{(i+1)}_{l,m} = \frac{1}{2f_cP} \Big(\frac{1}{\Gamma}\sumP x^{(i)}_{l,m,p} + \sumP \gamma_{n,l,m,p}\Big)$}
    \STATE {$t_{l,m}^{(i+1)} = \cQ_{\cT}(t^{(i+1)}_{l,m})$ \label{step5}}
    \FOR{$p=1:P$}
    \STATE{$x^{(i+1)}_{l,m,p} =  2f_ct^{(i)}_{l,m} - \gamma_{n,l,m,p}$ \label{step7}}
    \ENDFOR
    \IF{$\frac{\|\bx^{(i+1)}_{l,m} - \bx^{(i)}_{l,m}\|^2 + (t^{(i+1)}_{l,m}-t^{(i)}_{l,m})^2}{\|\bx^{(i+1)}_{l,m}\|^2 + (t^{(i+1)}_{l,m})^2} \leq \epsilon$}
    \RETURN{$t^{\star}_{l,m} = t^{(i+1)}_{l,m}$ and $x^{\star}_{l,m,p} = x^{(i+1)}_{l,m,p}$}, $\forall p$
    \ENDIF
    \STATE $i = i+1$
    \STATE{$t^{\star}_{l,m} = t^{(i)}_{l,m}$ and $x^{\star}_{l,m,p} = x^{(i)}_{l,m,p}$}, $\forall m,p$
    \ENDWHILE
    \ENDFOR
\end{algorithmic}
\end{algorithm}
\begin{figure*}[htp]
        \centering
        \subfloat[]{
        \includegraphics[scale=.33,trim=0.0cm 0cm 0.0cm 0cm, clip]{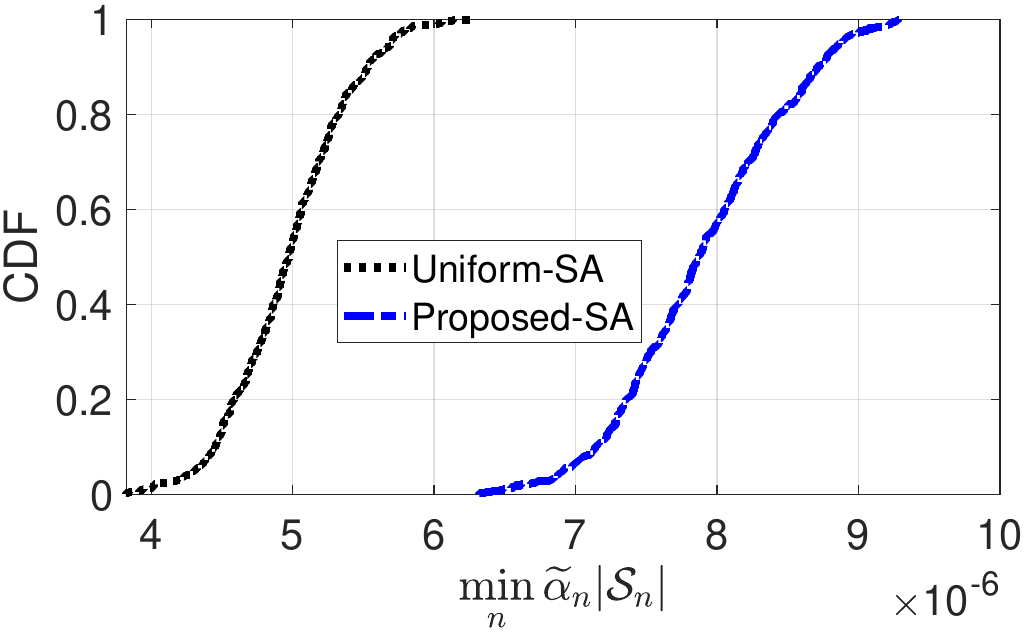}
        \label{FigPSAUSA}
        }
       \centering
        \subfloat[]{
        \includegraphics[scale=.325,trim=0cm 0cm 0.0cm 0cm, clip]{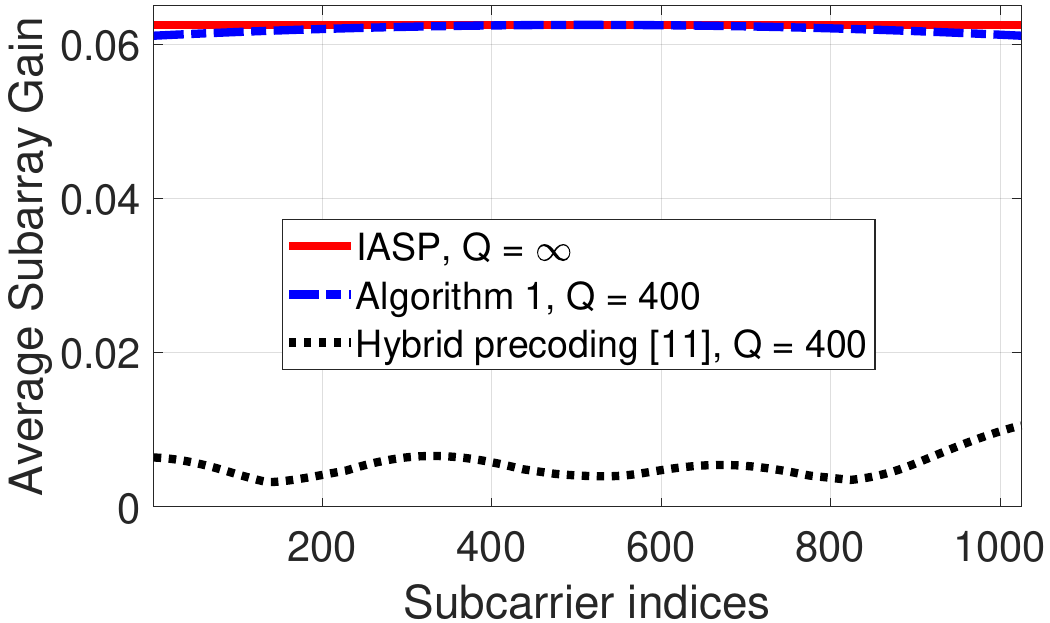}
        \label{FigAG}
        }
        \centering
        \subfloat[]{\includegraphics[scale=.34,trim=0.0cm 0cm 0.0cm 0cm, clip]{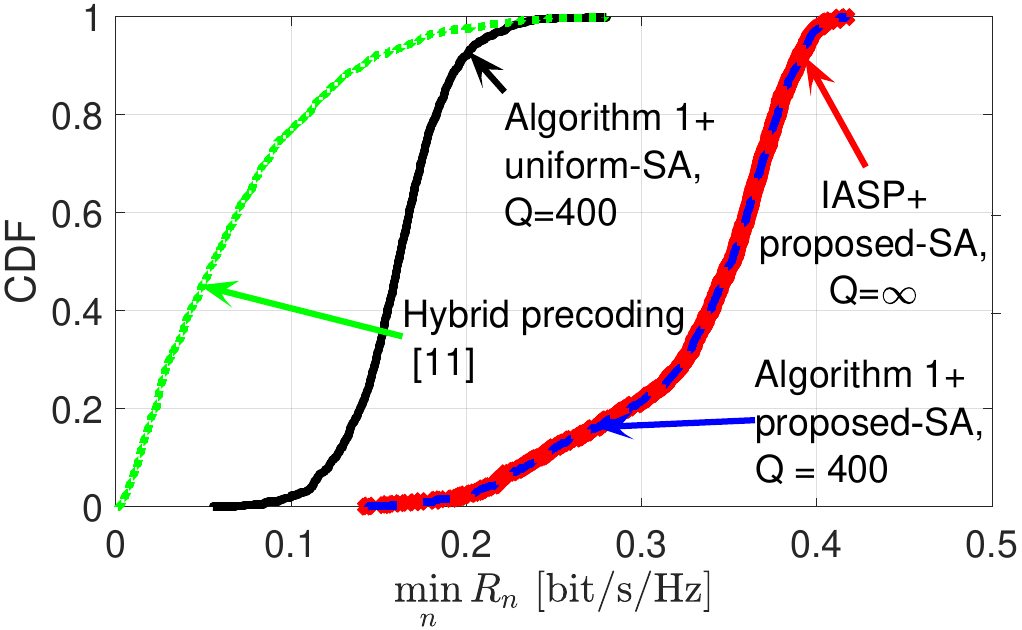}
         \label{Fig4}}  
        \caption{(a) CDFs of $\min_{n}\tilde{\alpha}_n|\cS_n|$ of the proposed-SA and uniform-SA. (b) Average subarray gain performances of Algorithm~\ref{algorithm1} with $Q = 400$, IASP with $Q = \infty$, and benchmarkc      \cite{Najjar2022} with $Q=400$. (c) CDFs of $\min_{n} R_n$ of the proposed approach with $Q=400$, IASP with $Q=\infty$, and benchmark \cite{Najjar2022} with $Q = 400$.}
        \vspace{-0.5cm}    
    \label{figSim}
    \end{figure*}
A formal description of the alternating optimization is presented in Algorithm \ref{algorithm1}. 
Encountering the discrete constraint in \eqref{eq:opt1b}, it finds an approximate solution $\{t_{l,m}^{\star}\}$ and $\{x_{l,m,p}^{\star}\}$ to \eqref{eq:opt1}. 
In Step~\ref{step1}, we set $t^{(1)}_{l,m} = 0$, $\forall m$ and $x^{(1)}_{l,m,p} = 0$, $\forall m,p$, as the initial TTD and PS values, respectively.  
At the $(i+1)$th iteration, we update $t^{(i+1)}_{l,m}$ based on \eqref{vartheta_lm}.   
Then, in Step~\ref{step5}, the TTD value $t^{(i+1)}_{l,m}$ is quantized by $\cQ_{\cT}(\cdot)$, where $\cQ_{\cT}(x) = \tilde{t}$ such that $\tilde{t} = \argmin_{t\in \cT} |x-t|$.  
In Step~\ref{step7}, the $\{x^{(i+1)}_{l,m,p}\}$ is updated based on \eqref{eq:x_lmp}.
These steps are repeated until either the convergence condition is satisfied or the maximum number of iterations $N_{iter}$ is reached.
The convergence condition is met when the normalized mean square error (NMSE) at the $(i+1)$ iteration is upper bounded by $\epsilon$, i.e.,  $\frac{\|\bx^{(i+1)}_{l,m} - \bx^{(i)}_{l,m}\|^2 + (t^{(i+1)}_{l,m}-t^{(i)}_{l,m})^2}{\|\bx^{(i+1)}_{l,m}\|^2 + (t^{(i+1)}_{l,m})^2}\leq \epsilon$. 

\section{Simulation Results}
This section presents numerical simulations that evaluate the cumulative distribution function (CDF) of the asymptotic min-subarray gain $\min_{n} \tilde{\alpha}_n |\cS_n|$, the average subarray gain, and the CDF of the minimum achievable rate $\min_{n} R_n$ of the proposed approach.\looseness=-1 
The channel matrices is generated based on the sub-THz channel model in \cite{Han2015, Jornet2011}. 
We set the system parameters to $f_c = 300$ GHz, $B = 30$ GHz, $K = 1025$, $N = 4$, $N_{RF} = 16$, $M = 16$, $N_t = 1024$, $N_r = 4$, and $SNR = \rho = 10$ (dB). 
The communication distances from the transmitter to the four users are, respectively, $\{10,15,20,25\}$ meters. 
The channel path gains $\{\alpha_{k,n}\}$ and delays $\{\tau_n\}$ are randomly generated following \cite{Han2015, Jornet2011}.
The AoAs and AoDs are uniformly drawn in $[-\frac{\pi}{2},\frac{\pi}{2}]$.
The maximum TTD value is $(Q-1)\tau = 1596$ ps, where $Q=400$ and $\tau = 4$ ps.
We set $N_{iter} = 50$ and $\epsilon = 0.01$ in Algorithm~\ref{algorithm1}. 
The performance curves are averaged over 500 channel realizations. \looseness=-1 
Fig.~\ref{FigPSAUSA} demonstrates the CDF of the asymptotic min-subarray gain $\min_{n}\tilde{\alpha}_n|\cS_n|$ performance of the proposed subarray allocation (proposed-SA) approach. 
The uniform subarray allocation (uniform-SA)  in Fig.~\ref{FigPSAUSA} allocates $|\cS_n|= 4$, for $n=1,\dots,4$. 
Fig.~\ref{FigPSAUSA} illustrates a substantial improvement of the proposed approach compared to the uniform-SA. 
For example, the proposed-SA allocates $|\cS_1|=2$, $|\cS_2|=3$, $|\cS_3| = 5$, and $|\cS_4|=6$ because the channel in Fig.~\ref{FigPSAUSA} is realized such that the third and fourth users experience more severe path-loss than the first and second users. 
This confirms the observation in Remark~\ref{rmk1}. \looseness=-1 

The average subarray gain at the $k$th subcarrier is defined as $\frac{1}{N_{RF}}\sum_{n=1}^{N}\sum_{l \in \cS_n}|\bu^{\dagger}_{k,n,l}\bff_{k,l}|$, where $\bu_{k,n,l}$ and $\bff_{k,l}$ are defined in \eqref{eq:subarrayRV} and \eqref{eq:subprecoder}, respectively.
Fig.~\ref{FigAG} compares the average subarray gain performances of Algorithm~\ref{algorithm1} with $Q=400$, benchmark \cite{Najjar2022} with $Q = 400$, and IASP in \eqref{eq:assumption1} with $Q = \infty$. 
It can be observed in Fig.~\ref{FigAG} that the average subarray gain of Algorithm~\ref{algorithm1} with the finite TTD resolution $Q=400$ shows a similar performance as that of the IASP with $Q = \infty$. 
Additionally, the Algorithm~\ref{algorithm1} achieves higher average subarray gain performance than the benchmark \cite{Najjar2022} with the same TTD resolution. \looseness=-1

Fig~\ref{Fig4} illustrates the CDFs of the minimum achievable rate ($\min_{n} R_n$) of the proposed Algorithm~\ref{algorithm1} with the uniform-SA and the proposed-SA compared with the benchmark \cite{Najjar2022}.  
When $Q = 400$, it is seen from Fig.~\ref{Fig4} that the proposed-SA outperforms the uniform-SA substantially.
This is consistent with the asymptotic min-subarray gain trend shown in Fig.~\ref{FigPSAUSA}.
The Algorithm~\ref{algorithm1} achieves a close performance to the IASP in \eqref{eq:assumption1} with $Q = \infty$ and substantially outperforms the benchmark \cite{Najjar2022}, which is aligned with the subarray gain trend in Fig. \ref{FigAG}.
\section{Conclusion}
A signal processing approach has been proposed to ensure a fair rate performance of a real-world-suitable wideband sub-THz massive MIMO OFDM system that employs no analog combiners. 
The no-combiner architecture is critical in practical deployment due to space limitation and power loss. 
First, a subarray allocation problem was formulated to provide a fair allocation to each user when the number of subarrays is larger than the number of users. 
An approximate solution to the first sub-problem was obtained by continuous relaxation, which suggests that the allocation of subarrays to each user should be inversely proportional to its channel gain. 
Second, a joint PS and TTD sub-precoders design problem was formulated to maximize the subarray gain subject to the limited range and finite TTD resolution.
A solution was attained by transforming it to the phase domain and applying an alternating optimization technique, which achieves a close subarray gain performance to the ideal analog sub-precoders that maximize the subarray gain.  
The advantages of the proposed approach were demonstrated by numerical results.
\looseness=-1 
\appendices
\section{Proof of Lemma~\ref{lm:RateUB}}\label{appendixC}
Plugging $\bH_{k,n}$ in \eqref{eq:channel} into $R_n$ in \eqref{eq:Rate}  gives 
\vspace{-0.2cm}
\begin{subequations}
\label{eq:RateBound}
\beq
\d4\d4\d4 \scalemath{0.9}{R_n} \d4&\leq&\d4 \scalemath{0.8}{\sumK \log_2\Big( 1 + \frac{\rho N_rN_t}{N}\|\alpha_{k,n}\bv_{k,n}\bu^{\dagger}_{k,n}\bF_1\bF_{2,k}\bw_{k,n}\|^2\Big)}, \label{eq:RateBounda}\\
\d4&\leq&\d4 \scalemath{0.9}{\frac{\rho N_r N_t}{\ln(2)N}\sumK|\alpha_{k,n}|^2\|\bv_{k,n}\|^2\|\bu^{\dagger}_{k,n}\bF_1\bF_{2,k}\|^2\|\bw_{k,n}\|^2},\label{eq:RateBoundb}\\
\d4&\leq&\d4 \scalemath{0.9}{\frac{\omega\rho N_r N_t}{\ln(2)N}\sumK|\alpha_{k,n}|^2\|\bu^{\dagger}_{k,n}\bF_1\bF_{2,k}\|^2}, \label{eq:RateBoundc}
\eeq 
\end{subequations}
where \eqref{eq:RateBounda} follows from removing the inter-user interference term in $R_n$, \eqref{eq:RateBoundb} is due to the fact that $\ln(1+x) \leq x$ for $x \geq 0$ and applying Cauchy-Schwartz's inequality to 
$\scalemath{0.9}{\|\bv_{k,n}\bu^{\dagger}_{k,n}\bF_1\bF_{2,k}\bw_{k,n}\|^2}$, and \eqref{eq:RateBoundc} follows from \eqref{eq:arrayRV2} and \eqref{eq:assumption2}.
We note that  
\begin{subequations}
\label{eq:UB}
\beq
\d4\d4\d4\scalemath{0.9}{\|\bu^{\dagger}_{k,n}\bF_1\bF_{2,k}\|^2} \d4&=&\d4 \scalemath{0.9}{\sum_{l \in \cS_n}\|\bu^{\dagger}_{k,n,l}\bff_{k,l}\|^2 +\d4 \sum\limits_{\substack{l \in \cS_{n'}, n' \neq n}}\d4\d4\|\bu^{\dagger}_{k,n,l}\bff_{k,l}\|^2},\label{eq:UBa}\\
\d4&=&\d4 \scalemath{0.9}{\sum_{l \in \cS_n}\!\!\|\bu^{\dagger}_{k,n,l}\bu_{k,n,l}\|^2 + \d4\d4\sum\limits_{\substack{l \in \cS_{n'}, n' \neq n}}\d4\d4\!\|\bu^{\dagger}_{k,n,l}\bu_{k,n',l}\|^2\!\!},\label{eq:UBb}\\ 
\d4&=&\d4 \scalemath{0.9}{\frac{|\cS_n|}{N_{RF}} + \d4\d4 \sum\limits_{\substack{l \in \cS_{n'}, n' \neq n}}\|\bu^{\dagger}_{k,n,l}\bu_{k,n',l}\|^2},\label{eq:UBc}
\eeq
\end{subequations}
where \eqref{eq:UBa} follows from \eqref{eq:Fk}, \eqref{eq:UBb} is due to \eqref{eq:assumption1}, and \eqref{eq:UBc} is due to \eqref{eq:subarrayRV}. 
For $n \neq n'$, the subarray response vectors $\bu_{k,n,l}$ and $\bu_{k,n',l}$ are orthogonal with probability $1$, as \small$MP \rightarrow \infty$\normalsize \cite{ayach2014}.  
Hence,
the r.h.s of \eqref{eq:RateBoundc} tends to $\frac{\omega\rho N_r N_t}{\ln(2)N}\frac{|\cS_n|}{N_{RF}}\sumK|\alpha_{k,n}|^2$ with probability 1, as \small$MP \rightarrow \infty$\normalsize, completing the proof.\looseness=-1 
\section{Proof of Lemma \ref{lemma:optimal_SA2}}
\label{appendixA}
Note that 
\small$\!\! \min_{n} \widetilde{\alpha}_n |\cS_n| = \!\! \Big(\!\!\min_{n} \frac{|\cS_n|}{1/\widetilde{\alpha}_n}\!\!\Big) \Big(\!\!\sumN \frac{1/\widetilde{\alpha}_n}{\sumN 1/\widetilde{\alpha}_n} \!\!\Big) \!\! \leq \frac{|\cS_1|}{1/\widetilde{\alpha}_1}\frac{1/\widetilde{\alpha}_1}{\sumN 1/\widetilde{\alpha}_n} + \dots + \!\! \frac{|\cS_N|}{1/\widetilde{\alpha}_N}\frac{1/\widetilde{\alpha}_N}{\sumN 1/\widetilde{\alpha}_n}\!\!=\frac{\sumN |\cS_n|}{\sumN 1/\widetilde{\alpha}_n},$\normalsize~implying that \small$\max_{\{|\cS_n|\}} \min_{n} \widetilde{\alpha}_n|\cS_n| \leq \frac{\sumN |\cS_n|}{\sumN 1/\widetilde{\alpha}_n}.$\normalsize~ \sloppy  
The upper bound is achieved as the equality if and only if   $\scalemath{0.8}{|\cS_n| = \frac{(1/\widetilde{\alpha}_n) \sum_{n'=1}^{N} |\cS_{n'}|}{\sumN 1/\widetilde{\alpha}_n} = \frac{N_{RF}/\widetilde{\alpha}_n}{\sumN 1/\widetilde{\alpha}_n}},$ completing the proof.\sloppy 
\vspace{-0.1cm}
\bibliographystyle{IEEEtran} 
\vspace{0.4cm}
\bibliography{biblib} 
\end{document}